\documentclass[conference,a4paper]{IEEEtran}

\newif\ifthreevalue
\threevaluetrue

\newif\ifarXiv
\arXivtrue

\usepackage[utf8]{inputenc} 
\usepackage[T1]{fontenc}
\usepackage{url}              

\usepackage[cmex10]{amsmath}  
\usepackage{mleftright}       
\mleftright                   

\usepackage{graphicx}         
\usepackage{booktabs}         

\usepackage{amsmath,amssymb,amsfonts,amsthm}
\usepackage{enumitem}
\usepackage{mathtools}
\usepackage{algorithmic}
\usepackage{textcomp}
\usepackage[table,xcdraw]{xcolor}
\usepackage{tikz}

\usepackage{subcaption}
\usepackage{multicol,multirow}
\usepackage{comment}

\usepackage[style=ieee,citestyle=numeric-comp,sorting=nty]{biblatex}
\newtheorem{theorem}{Theorem}

\newtheorem{proposition}{Proposition}

\newtheorem{lemma}{Lemma}
\newtheorem{definition}{Definition}

\newcommand{\bF}{\mathbb{F}}

\newcommand{\bR}{\mathbb{R}}

\newcommand{\bZ}{\mathbb{Z}}

\newcommand{\cA}{\mathcal{A}}

\newcommand{\cN}{\mathcal{N}}

\newcommand{\bolda}{\mathbf{a}}
\newcommand{\boldb}{\mathbf{b}}
\newcommand{\boldc}{\mathbf{c}}

\newcommand{\boldg}{\mathbf{g}}

\newcommand{\boldq}{\mathbf{q}}
\newcommand{\boldr}{\mathbf{r}}

\newcommand{\boldu}{\mathbf{u}}
\newcommand{\boldv}{\mathbf{v}}
\newcommand{\boldw}{\mathbf{w}}
\newcommand{\boldx}{\mathbf{x}}
\newcommand{\boldy}{\mathbf{y}}

\newcommand{\boldalpha}{\boldsymbol{\alpha}}
\newcommand{\boldbeta}{\boldsymbol{\beta}}

\newcommand{\bolddelta}{\boldsymbol{\delta}}

\DeclareMathOperator{\RSpan}{\operatorname{Span}_{\bR}}
\DeclarePairedDelimiter{\floor}{\lfloor}{\rfloor} 
\DeclarePairedDelimiter{\Norm}{\lVert}{\rVert} 

\newcommand{\norm}[1]{\Norm{#1}}
\newcommand{\normscaled}[1]{\Norm*{#1}}

\DeclareSymbolFont{bbold}{U}{bbold}{m}{n}
\DeclareSymbolFontAlphabet{\mathbbold}{bbold}
\newcommand{\1}{\mathbbold{1}}

\excludecomment{appendix}

\begin{document}

\title{Approximate Private Inference in Quantized Models}

\author{
  \IEEEauthorblockN{Zirui Deng}
  \IEEEauthorblockA{
    \textit{Washington University in St. Louis}\\
    d.ken@wustl.edu}
  \and
  \IEEEauthorblockN{Netanel Raviv}
  \IEEEauthorblockA{
    \textit{Washington University in St. Louis}\\
    netanel.raviv@wustl.edu}
}


\maketitle

\begin{abstract}
    \textit{Private inference} refers to a two-party setting in which one has a model (e.g., a linear classifier), the other has data, and the model is to be applied over the data while safeguarding the privacy of both parties. In particular, models in which the weights are \textit{quantized} (e.g., to~$\pm1$) gained increasing attention lately, due to their benefits in efficient, private, or robust computations.

    Traditionally, private inference has been studied from a cryptographic standpoint, which suffers from high complexity and degraded accuracy. More recently, Raviv \textit{et al.} showed that in quantized models, an information theoretic tradeoff exists between the privacy of the parties, and a scheme based on a combination of Boolean and real-valued algebra was presented which attains that tradeoff. Both the scheme and the respective bound required the computation to be done \textit{exactly}.
        
    In this work we show that by relaxing the requirement for exact computation, one can break the information theoretic privacy barrier of Raviv \textit{et al.}, and provide better privacy at the same communication costs. We provide a scheme for such approximate computation, bound its error, show its improved privacy, and devise a respective lower bound for some parameter regimes.
\end{abstract}

\begin{IEEEkeywords}
Information-theoretic privacy, private computation.
\end{IEEEkeywords}

\pagestyle{plain}

\section{Introduction}
Information privacy issues affect various aspects of daily life, from financial security to social media, as users consistently rely on service providers for data and computation management. In particular, the inference phase of machine learning raises privacy concerns. In this setting, a server (e.g., a service provider) holds a trained model (e.g., a neural network) and offers its use to users at charge. A user, wanting to utilize the model, exchanges information with the server for effective inference on input data. A private inference protocol is one which guarantees some level of privacy for both parties.

Private inference has been studied extensively from a cryptographic perspective (e.g.,~\cite{gilad2016cryptonets,juvekar2018gazelle,mohassel2017secureml,reagen2021cheetah,boemer2020mp2ml}). More recently, inspired by the \textit{Private Information Retrieval} literature, Ref.~\cite{raviv2022information} viewed private inference as a task of retrieving an inner product between a data vector held by a user and a parameter vector held by a server, both real-valued; this is a preliminary step in most machine learning models.

Specifically, \cite{raviv2022information} focuses on \textit{quantized} models, in which the parameter vector is restricted to a finite set of real values. Such models gained increasing attention lately due to their benefits in hardware and distributed implementations, as well as being robust and efficiently trainable~\cite{teerapittayanon2017distributed,hubara2016binarized,mcdanel2017embedded,raviv2021enhancing}. In their scheme, the server publishes a bit vector which describes its parameter vector in a privacy-preserving way. Based on this public vector the user computes several real vectors, and sends the respective inner products with the data to the server. The server then extracts the required inner product, completes the computation of the model locally, and sends the result back to the user. Inherent limitations of this approach include the ability of the user to train a model using black-box queries, and are discussed in~\cite{raviv2022information}.

The merit of a given scheme is interplay between its publication cost (number of bits published by the server), user-privacy (the number of dimensions on which the data is revealed), and server-privacy (the mutual information between the parameter vector and its public version). Ref.~\cite{raviv2022information} proves a lower bound for the amount of information leakage from both the server and the user combined, and gives a scheme which achieves it with equality. In this work we show that the leakage can be reduced if we are allowed an \textit{approximate} computation of the inner product, a reasonable compromise for most use cases in the inherently inaccurate field of machine learning. Our scheme applies to all parameters but requires Hadamard matrices of a certain order.


This paper is structured as follows. Problem setup and previous results are given in Section~\ref{section:preliminaries}, along with some mathematical background. Our scheme for approximate private inference and its guarantees in terms of privacy and error are given in Section~\ref{section:breaking}. In Section~\ref{section:lowerbound} we prove a lower bound for server-privacy in the case of \textit{linear} decoding and compare it to the one derived in Section~\ref{section:breaking}.

\section{Preliminaries}\label{section:preliminaries}
\subsection{Problem statement and Previous work}
The problem setup involves a server who holds a weight vector $\textbf{w} \in \{\pm1 \}^n$, seen as randomly chosen from a uniform distribution $W = \textrm{Unif}(\{\pm1 \}^n)$, and a user who holds a data vector $\textbf{x} \in \bR^n$ that is randomly chosen from some continuous data
distribution~$X$. The server would like to retrieve the signal~$\textbf{wx}^\intercal$ while guaranteeing some level of privacy for both parties. Retrieving the signal is a preliminary step in many machine learning paradigms, e.g., linear classification/regression, logistic regression, neural networks, and more.

We focus on protocols of the query-answer form, where the server sends the user a query $\textbf{q} \in \{\pm1\}^d$ for some $d$, randomly chosen from a distribution $Q$ that is a deterministic function of $W$,
i.e., $H(Q|W) = 0$, where $H$ is the entropy function. The user then computes~$\ell$ vectors $\textbf{v}_1, \ldots, \textbf{v}_\ell \in \bR^n$ from $\textbf{q}$, and sends an answer $\mathcal{A} = \{\textbf{v}_i\textbf{x}^\intercal\}_{i=1}^{\ell}$ to
the server. Finally, the server combines the answers to retrieve~$\boldw\boldx^\intercal$. The merit of such protocols is measured by the following quantities. 

\begin{itemize}
    \item \textbf{Publication cost}, i.e., the number of bits~$d$. The structure of the problem allows the server to send the same~$\textbf{q}$ to multiple users without any loss of privacy, e.g., in the manner of a public forum, and hence the term.
    \item \textbf{Server-privacy}, measured by the mutual information $I(W; Q)$.
    \item \textbf{User-privacy}, measured by the dimension of the subspace on which $\textbf{x}$ is revealed, i.e., the parameter~$\ell$. Justification for using this metric has been given in~\cite{raviv2022information} by considering the independent components of $X$. 
\end{itemize}

For a privacy parameter~$t$ ($1\le t\le n$), ~\cite{raviv2022information} provides a private inference scheme that has publication cost $d = n - t$, server-privacy $I(W; Q) = n - t$, and user-privacy $\ell = t$. The scheme has proved optimal in terms of privacy guarantees and publication cost, thanks to this main result that demonstrates a fundamental tradeoff in the two privacy measurements:

\begin{theorem}\cite[Thm.~2]{raviv2022information}
\label{theorem:privacyTradeoff}
    $I(W; Q) + \ell \ge n$.
\end{theorem}

\subsection{Mathematical background}\label{section:background}

We refer to~$\bF_2$ as the binary field in its $\{\pm 1\}$ representation, i.e., $-1$ represents the Boolean ``one,'' and $1$ represents
the Boolean ``zero.'' We use $\oplus, \odot$ for operations in $\bF_2$ and $+, \cdot$ for
operations in $\bR$. For a set~$S\subseteq [n]$, where~$[n]\triangleq\{1,2,\ldots,n\}$, we let~$\1_S\in\bF_2^n$ be the characteristic vector of~$S$, i.e., $(\1_S)_j=-1$ if~$j\in S$, and~$(\1_S)_j=1$ otherwise. We use~$\otimes$ to denote Kronecker's delta product between matrices, use~$w_H,d_H$ to denote Hamming weight and Hamming distance, respectively, and use~$\norm{\cdot}_2$, $d_2$ for the respective Euclidean quantities.

Let $V$ be a subspace of $\bF_2^n$. A coset of~$V$ is $V\oplus\bolda\triangleq\{\boldv \oplus \bolda|\boldv \in V\}$ for some $\bolda \in \bF_2^n$, and~$\bolda$ is said to be a \textit{shift vector} which defines the coset. Recall that any~$\boldu\in V\oplus\bolda$ can be used as a shift vector, i.e., $V\oplus \boldu=V\oplus\bolda$ for every~$\boldu\in V\oplus\bolda$. Furthermore, any two cosets are of equal size and are disjoint, and the family of all cosets of~$V$ covers all of~$\bF_2^n$.

Fixing any full-rank matrix~$M\in\bF_2^{(n-\dim V)\times n}$ whose right kernel is~$V$ gives rise to the definition of \textit{syndromes}. A syndrome of a vector~$\boldu\in\bF_2^n$ is the vector~$M\odot \boldu^\intercal\in\bF_2^{n-\dim V}$. It is known that vectors have the same syndrome if and only if they belong to the same coset.


\section{Breaking the Information-Theoretic Barrier by Approximation}\label{section:breaking}
Theorem~\ref{theorem:privacyTradeoff} provides robust limitations for the amount of leaked information from both parties of the protocol, where the product~$\boldx\boldw^\intercal$ should be retrieved \textit{exactly}. In this section it is shown that by settling for an approximation thereof, one can lower the information leakage from~$\boldw$. For a tunable privacy parameter~$t$, the scheme applies whenever there exists a Hadamard matrix of order\footnote{Hadamard matrices are conjectured to exist whenever~$4\vert t$, and the conjecture is verified for almost all~$t$'s up to~$2000$~\cite{djokovic2014some}. For~$t$'s which are an integer power of two, the \textit{Sylvester construction} provides an easy recursive way of constructing Hadamard matrices.}~$t$. We begin with an overview of the scheme in~\cite{raviv2022information}, an intuitive explanation about how we extended it to the approximate setting, followed by technical lemmas, and conclude with presenting the algorithm and analyzing its guarantees.

\subsection{Overview}
For any~$t\in[n]$, the scheme in~\cite{raviv2022information} relies on a fixed partition~$S_1,\ldots,S_t$ of~$[n]$ (i.e., the~$S_i$'s are disjoint and their union is~$[n]$), a Boolean subspace~$V=\operatorname{Span}_{\bF_2}\{\1_{S_i}\}_{i=1}^t$ in~$\bF_2^n$ which follows from the partition, and a fixed matrix $L\in\{\pm1\}^{t\times t}$ that is invertible over~$\bR$. 

  In the first step, the server communicates the identity of a coset of~$V$ which contains~$\boldw$ to the user by computing a syndrome, which in turn identifies a shift vector~$\boldu$ that defines that coset (i.e., $\boldw\in V\oplus \boldu$). Then, the user defines vectors~$\boldv_i=\boldu\oplus \boldc_i$ for~$i\in[t]$, where~$\boldc_i=\bigoplus_{r=1}^tL_{i,r}\odot \1_{S_r}$ for each~$i$. The vectors~$\boldv_i$ are used to send the response~$\{\boldx\boldv_i^\intercal\}_{i=1}^t$, which enables the server to compute~$\boldw\boldx^\intercal$ since~$\boldw\in\RSpan\{\boldv_i\}_{i=1}^t$; the latter fact follows from the choice of the vectors~$\boldc_i$, and the full proof from~\cite{raviv2022information} is omitted due to space constraints.


The approximation scheme devised herein is based on retrieving~$\boldx\boldw'^\intercal$ for some~$\boldw'\in V\oplus \boldu'$ with~$V\oplus\boldu\ne V\oplus\boldu'$, i.e., a~$\boldw'$ which lies in a \textit{different} coset of~$V$ than the one containing~$\boldw$. Namely, the server will communicate a coset of~$V$, again by computing its syndrome, that is in some sense adjacent to the one containing~$\boldw$. This enables the server to randomize over the choice of that coset, thus adding uncertainty to~$W|Q$ from the perspective of the user (i.e., increasing~$H(W|Q)$, and therefore reducing~$I(W;Q)$). The user will then define~$\boldv_i'=\boldu'\oplus \boldc_i$ (for the same~$\boldc_i$'s as above), and similarly respond with~$\{ \boldx\boldv_i'^\intercal\}_{i=1}^t $. Exact computation of~$\boldw\boldx^\intercal$, however, will not be possible. The quality of the approximation depends on an upper bound on~$|\boldx\boldw^\intercal-\boldx\boldw'^\intercal|$, which will follow from certain Hamming distance properties of vectors in~$V\oplus\boldu'$, and will be presented in the next section.

\subsection{Technical lemmas}
In what follows we rely on the tools from~\cite{raviv2022information} mentioned in the above overview, namely the partition~$S_1,\ldots,S_t$, the matrix~$L$, and the resulting vectors~$\boldc_i$. We also use~$\boldu$ as a shift vector such that~$\boldw\in V\oplus\boldu$, and use~$\boldu'$ as an alternative shift vector that will be specified later. First, our approach relies on the following lemma, which easily follows from the Cauchy-Schwartz inequality.
\begin{lemma}\label{lemma:CauchySchwartz}
	If~$\norm{\boldx}_2\le c$, then~$|\boldx\boldw^\intercal-\boldx\boldw'^\intercal|\le c\cdot \norm{\boldw-\boldw'}_2$ for every~$\boldw'$.
\end{lemma}
This lemma implies that given~$\{ \boldx\boldv_i'^\intercal \}_{i=1}^t$ for~$\{\boldv_i'=\boldu'\oplus\boldc_i\}_{i=1}^t$ as above, a promising strategy for approximating~$\boldx\boldw^\intercal$ is finding~$\boldw'\in\RSpan\{\boldv_i'\}_{i=1}^t$ closest to~$\boldw$ in Euclidean distance, and output~$\boldx\boldw'^\intercal$. Note that the requested quantity~$\boldx\boldw^\intercal$ can potentially grow with~$n$ in spite of having bounded~$\norm{\boldx}_2$ (e.g.~$\boldx=\tfrac{1}{\sqrt{n}}\1$ and~$\boldw=\1$). Hence, providing a constant approximation (i.e., constant $\norm{\boldw-\boldw'}_2$) is highly valuable, especially for models which are not very sensitive to fluctuations in~$\boldx\boldw^\intercal$ (such as classification or logistic regression).

In order to find the closest~$\boldw'$  to~$\boldw$ in~$\RSpan\{\boldv_i'\}_{i=1}^t$, let~$\boldbeta=(\beta_1,\ldots,\beta_t)\in\bR^t$ be such that~$\boldw=\sum_{i=1}^{t}\beta_i\boldv_i$, which exist since~$\boldw\in \operatorname{Span}_{\bR}\{\boldv_i\}_{i=1}^t$ as mentioned earlier (see~\cite{raviv2022information}). Since~$\boldv_i=\boldu\oplus \boldc_i$ for every~$i$, it follows that~$\boldw$ can be written as
\begin{align}\label{equation:wRealCoefficients}
	\boldw=\sum_{i=1}^{t}\beta_i\boldc_i U=\boldbeta C U,
\end{align}
where~$U=\operatorname{diag}(\boldu)$ (i.e.,~$\boldu$ on the main diagonal, and zero elsewhere) and~$C\in\{\pm 1\}^{t\times n}$ is a matrix whose rows are the~$\boldc_i$'s. Similarly, any vector in~$\RSpan\{ \boldv_i' \}_{i=1}^t$ can be written as~$\sum_{i=1}^t\alpha_i\boldc_iU'=\boldalpha CU'$ for some~$\boldalpha=(\alpha_1,\ldots,\alpha_t)$, where~$U'=\operatorname{diag}(\boldu')$. Hence, we wish to find the vector~$\boldalpha$ which minimizes the function
\begin{align}\label{equation:DistanceSquareFunction}
	f_{\boldu'}(\boldalpha)&=\normscaled{\boldalpha CU'-\boldbeta CU}_2^2\nonumber\\
	&= \boldalpha CC^\intercal\boldalpha^\intercal+\boldbeta CC^\intercal\boldbeta^\intercal-2\boldalpha(CUU'C^\intercal)\boldbeta^\intercal.
\end{align}
To find the minimum of~$f_{\boldu'}$, we apply standard analytic methods, which begin with computing the gradient of~$f$.
\begin{align}\label{equation:fGradient}
	\nabla_{\boldalpha} f_{\boldu'}(\boldalpha)=2\boldalpha CC^\intercal-2\boldbeta CUU'C^\intercal.
\end{align}
To study the roots of this expression, we use Kronecker's delta product~$\otimes$ over~$\bR$. 
\begin{lemma}\label{lemma:Kronecker}
	The matrix~$C$ whose rows are the vectors~$\{ \boldc_i=\bigoplus_{r=1}^t L_{i,r}\odot\1_{S_r} \}_{i=1}^t$ can be written as~$C=L\otimes \1_{n/t}$, where~$\otimes$ is Kronecker's delta product, and~$\1_{n/t}$ is a vector of~$n/t$ $1$'s.
\end{lemma}

\begin{proof}
    See \ifarXiv the appendix\else~\cite{arXiv}\fi.
\end{proof}

Therefore, going back to~\eqref{equation:fGradient}, we have that whenever~$L$ is a Hadamard matrix of order~$t$ (i.e.,~$LL^\intercal=tI_t$, where~$I_t$ is the ~$t\times t$ identity matrix), it follows that
\begin{align}\label{equation:BBt}
	CC^\intercal&=(L\otimes \1_{n/t}) \cdot (L\otimes \1_{n/t})^\intercal\nonumber\\
 &=(L\otimes \1_{n/t}) \cdot (L^\intercal\otimes \1_{n/t}^\intercal) = (LL^\intercal)\otimes (\1_{n/t}\1_{n/t}^\intercal)\nonumber\\
	&= (tI_t)\otimes (n/t)=nI_t,
 \end{align}
and also,
\begin{align}\label{equation:BBt2}
C^\intercal C&= (L\otimes \1_{n/t})^\intercal \cdot (L\otimes \1_{n/t})\nonumber\\
 &=(L^\intercal\otimes \1_{n/t}^\intercal) \cdot (L\otimes \1_{n/t})=(L^\intercal L)\otimes (\1_{n/t}^\intercal\1_{n/t})\nonumber\\
	&= (tI_t)\otimes (J_{n/t})=t\cdot \operatorname{diag}(J_{n/t},\ldots,J_{n/t})
\end{align}
where~$J_{n/t}$ is an~$n/t \times n/t$ matrix of~$1$'s, and the various steps follow from the properties of Kronecker's delta product. Plugging~\eqref{equation:BBt} back to~\eqref{equation:fGradient}, it follows that~$f_{\boldu'}$ is minimized by
\begin{align}\label{equation:alphamin}
	\boldalpha_{\boldu'}^\text{min}=(\alpha_{\boldu',i}^\text{min})_{i=1}^n\triangleq\tfrac{1}{n}\boldbeta C UU' C^\intercal\overset{\eqref{equation:wRealCoefficients}}{=}\tfrac{1}{n}\boldw U'C^\intercal.
\end{align}
Hence, we have proved the following.
\begin{lemma}\label{lemma:closestVector}
	For a given~$\boldu'\in\{ \pm 1 \}^n$, the~$\boldw'$ which is closest to~$\boldw$ in~$\RSpan\{ \boldv_i'=\boldu'\oplus \boldc_i \}_{i=1}^t$ is 
	\begin{align*}
		\boldw'&=\sum_{i=1}^t\alpha^{\text{min}}_{\boldu',i}\boldc_i U'=\boldalpha_{\boldu'}^\text{min} C U'\\
  &\overset{\eqref{equation:alphamin}}{=}\tfrac{1}{n}\boldw U'C^\intercal CU'\overset{\eqref{equation:BBt2}}{=}\tfrac{t}{n}\boldw U'\operatorname{diag}(J_{n/t})U',
	\end{align*}
	where~$\operatorname{diag}(J_{n/t})$ is a shorthand notation for~$\operatorname{diag}(J_{n/t},\ldots,J_{n/t})$.
\end{lemma}
Hence, the resulting minimum distance between~$\boldw$ and~$\boldw'$ is as follows (see the complete derivation in \ifarXiv the appendix\else~\cite{arXiv}\fi).
\begin{align}\label{equation:ww'dist}
	&\norm{\boldw-\boldw'}_2^2= n-\tfrac{t}{n}(\boldw\oplus\boldu')\operatorname{diag}(J_{n/t})(\boldw\oplus\boldu')^\intercal.
\end{align}
To better understand~\eqref{equation:ww'dist}, we define the following notation. For a vector~$\bolddelta$ and~$p\in[t]$ let~$w_H(\bolddelta\vert_{S_p})$ be its Hamming weight (i.e., the number of~$-1$'s) in entries indexed by~$S_p$. 
The following is derived in \ifarXiv the appendix\else~\cite{arXiv}\fi.
\begin{align}\label{equation:fmincont}
\eqref{equation:ww'dist}	 =4d_H(\boldw,\boldu')-\tfrac{4t}{n}\sum_{p=1}^tw_H(\boldw\oplus\boldu'\vert_{S_p})^2.
\end{align}
By taking the square root of the resulting expression we have thus proved the following lemma.

\begin{lemma}\label{lemma:resultingApproximation}
	For a given~$\boldu'\in\{\pm 1\}^n$, the distance of the closest vector~$\boldw'$ to~$\boldw$ in~$\RSpan\{ \boldv_i'=\boldu'\oplus\boldc_i \}_{i=1}^t$ is 
	\begin{align*}
		\norm{\boldw-\boldw'}_2&=\sqrt{f_{\boldu'}(\boldalpha_{\textnormal{min}})}\\
  &=2\sqrt{d_H(\boldw,\boldu')-\tfrac{t}{n}\textstyle\sum_{p=1}^tw_H(\boldw\oplus\boldu'\vert_{S_p})^2}.
	\end{align*}
\end{lemma}


Since the subsequent algorithm depends on choosing a particular shift vector~$\boldu'$ to represent~$V\oplus \boldu'$, one might wonder if choosing a different shift vector~$\boldu''$ for the same coset (i.e., such that~$V\oplus \boldu'=V\oplus \boldu''$, see Section~\ref{section:background}) might alter the resulting closest vector~$\boldw'$ (Lemma~\ref{lemma:closestVector}). The following lemma demonstrates that this is not the case.

\begin{lemma}\label{lemma:upupp}
	For~$\boldu',\boldu''\in\bF_2^n$ such that~$\boldu'\oplus\boldu'' \in V$, we have $U'\operatorname{diag}(J_{n/t})U'=U''\operatorname{diag}(J_{n/t})U''$, where~$U'=\operatorname{diag}(\boldu')$ and~$U''=\operatorname{diag}(\boldu'')$. As a result, the choice of a shift vector to represent the coset~$V\oplus\boldu'$ does not change the resulting closest vector~$\boldw'$~(Lemma~\ref{lemma:closestVector}).
\end{lemma}
\begin{proof}
	Since~$\boldu'\oplus \boldu''\in V$, it follows that there exist a vector~$\boldv\in V$ such that~$\boldu'=\boldu''\oplus \boldv$, and thus~$U'=U''\cdot \operatorname{diag}(\boldv)$. Similarly, by the definition of~$V$, it follows that~$\boldv$ is of the form~$\boldv=\bigoplus_{\ell=1}^r \1_{S_{i_r}}$ for some~$i_1,\ldots,i_r\in[t]$, and thus
	\begin{align*}
		\operatorname{diag}(\boldv) = \prod_{\ell=1}^{r}\operatorname{diag}(\1_{S_{i_r}}).
	\end{align*}
	Consequently, 
	\begin{align*}
		&U'\operatorname{diag}(J_{n/t})U'=\\
        &=U''\prod_{\ell=1}^{r}\operatorname{diag}(\1_{S_{i_r}})\cdot \operatorname{diag}(J_{n/t})\cdot U''\prod_{\ell=1}^{r}\operatorname{diag}(\1_{S_{i_r}}).
	\end{align*}
	Since diagonal matrices commute, and since
	\begin{align*}
		&\operatorname{diag}(\1_{S_i})\operatorname{diag}(J_{n/t})\operatorname{diag}(\1_{S_i})=\\
        &=\operatorname{diag}(\1_{S_i})\operatorname{diag}(J_{n/t},\ldots,J_{n/t})\operatorname{diag}(\1_{S_i})\\
		&= \operatorname{diag}(J_{n/t},\ldots,\underbrace{-J_{n/t}}_{i\text{-th position}},\ldots,J_{n/t})\operatorname{diag}(\1_{S_i})=\operatorname{diag}(J_{n/t})
	\end{align*}
	for every~$i\in[t]$, it follows that~$U'\operatorname{diag}(J_{n/t})U'=U''\operatorname{diag}(J_{n/t}) U''$, which concludes the first statement of the lemma. The second statement follows from Lemma~\ref{lemma:closestVector}.
\end{proof}

\subsection{The algorithm}

In the upcoming algorithm, the server will communicate to the user a syndrome of a coset~$V\oplus \boldu'$ which contains a vector close to~$\boldw$ in Hamming distance. This coset will be chosen uniformly at random by selecting a shift vector from a set~$\Gamma\subseteq\{ \pm 1 \}^n$. Specifically, let~$\Gamma\subseteq\{ \pm  1\}^n$ be the set of all vectors~$\boldg$ such that~$w_H(\boldg)\le h$ and~$w_H(\boldg\vert_{S_i})<\frac{n}{2t}$ for every~$i\in[t]$, where~$h$ is an additional privacy parameter (notice that~$h\le n/2$ is a necessary condition). We begin by showing that every~$\boldg\in\Gamma$ defines a distinct coset of~$V$.

\begin{lemma}\label{lemma:cosetUniqueness}
	For distinct~$\boldg,\boldg'\in\Gamma$ we have~$V\oplus\boldg\cap V\oplus\boldg'=\varnothing$. Furthermore, we have that~$V\oplus(\boldr\oplus \boldg)\cap V\oplus(\boldr\oplus\boldg')=\varnothing$ for any vector~$\boldr$.
\end{lemma}
\begin{proof}
	Assume for contradiction that~$\boldg=\boldg'\oplus\boldv$ for some~$\boldv\in V$. By the definition of~$V$, we have that~$\boldv=\bigoplus_{\ell=1}^r \1_{S_{i_\ell}}$ for some~$r\in[t]$ and some~$S_{i_1},\ldots,S_{i_r}$. Hence, restricting our attention to entries in any one of these $S_{i_j}$'s, it follows that~$\boldg\vert _{S_{i_j}}=\boldg'\vert_{S_{i_j}}\oplus (-\1_{n/t})$, i.e., that~$\boldg\vert_{S_{i_j}}$ and~$\boldg'\vert_{S_{i_j}}$ are distinct in all~$n/t$ entries, which implies that~$w_H(\boldg\vert_{S_{i_j}})+w_H(\boldg'\vert_{S_{i_j}})=n/t$. However, by the definition of~$\Gamma$ we have that~$w_H(\boldg\vert_{S_{i_j}})+w_H(\boldg'\vert_{S_{i_j}})<2\cdot\frac{n}{2t}= n/t$, a contradiction. The ``furthermore'' part is similar.
\end{proof}

Having established the computation of the closest vector~$\boldw'$~(Lemma~\ref{lemma:closestVector}), the resulting approximation (Lemma~\ref{lemma:resultingApproximation}), the independence in the choice of the shift vector (Lemma~\ref{lemma:upupp}), and that the cosets~$\{ V\oplus(\boldw\oplus\boldg) \}_{\boldg\in\Gamma}$ are distinct (Lemma~\ref{lemma:cosetUniqueness}), we are now in a position to formulate and analyze an algorithm for \textit{approximate} information-theoretic private inference. In what follows, let $B(A, \boldb)$ be any deterministic algorithm that finds a solution~$\boldx$ to the equation $A\boldx = \boldb$ over~$\bF_2$; the algorithm~$B$ is assumed to be known to all. Also, let~$M\in\bF_2^{(n-t)\times n}$ be a matrix whose right kernel is~$V$.


\begin{enumerate}
	\item The server:
	\begin{enumerate}
		\item Chooses~$\boldg\in\Gamma$ uniformly at random.
		\item Publishes~$\boldq^\intercal\triangleq M\odot(\boldw\oplus\boldg)^\intercal\in\bF_2^{n-t}$.
	\end{enumerate}
	\item The user:
	\begin{enumerate}
		\item Defines~$\boldu'\triangleq B(M,\boldq)$. 
		\item Defines $\boldv_i'=\boldu' \oplus \boldc_i$, where~$\boldc_i=\bigoplus_{r=1}^t\left[ L_{i,r}\odot \1_{S_r} \right]$, for each~$i\in[t]$. 
		\item Sends~$\cA=\{ \boldv_i'\boldx^\intercal \}_{i=1}^t$ to the server.
	\end{enumerate}
	\item The server:
	\begin{enumerate}
		\item Computes~$\boldu'\triangleq B(M,\boldq)$. This is the same vector~$\boldu'$ that is found above since~$B$ is deterministic.
		\item Computes~$\boldalpha_{\boldu'}^\text{min}=\tfrac{1}{n}\boldw U'C^\intercal$, where~$U'=\operatorname{diag}(\boldu')$ and~$C\in\{\pm 1\}^{t\times n}$ is a matrix whose rows are the aforementioned~$\boldc_i$'s.
		\item Outputs~$\boldw'\boldx^\intercal\triangleq\sum_{i=1}^t \alpha_{\boldu',i}^\text{min}\boldv_i'\boldx^\intercal$.
	\end{enumerate}	
\end{enumerate}
It is evident that the publication cost equals~$n-t$ and the user-privacy equals~$t$, both identical to the respective ones in~\cite{raviv2022information}. The server-privacy, however, is improved, at the price of an approximate computation of~$\boldx\boldw^\intercal$ instead of an exact one. The next two theorems demonstrate the resulting approximation and improved privacy. The proof of the first one, given in \ifarXiv the appendix\else~\cite{arXiv}\fi, is based on solving a constrained optimization problem that is formulated using previous lemmas.

\begin{theorem}\label{theorem:errorbound}
	In the above algorithm, we have	$|\boldw'\boldx^\intercal-\boldw\boldx^\intercal|\le 2\norm{\boldx}_2\cdot\sqrt{h(1-h/n)}$.
\end{theorem}

To interpret Theorem~\ref{theorem:errorbound}, observe that the quantity~$\boldw\boldx^\intercal$ to be computed satisfies $| \boldw\boldx^\intercal| \le \norm{\boldw}_2\norm{\boldx}_2 = \sqrt{n}\norm{\boldx}_2$ by the Cauchy-Schwartz inequality; that is, it lies in a range of length~$2\sqrt{n}\norm{\boldx}_2$. Hence, the error implied by Theorem~\ref{theorem:errorbound}, relative to this range, is
\begin{align}\label{equation:relativeerror}
    \frac{|\boldw'\boldx^\intercal-\boldw\boldx^\intercal|}{2\sqrt{n}\norm{\boldx}_2}\le \frac{\sqrt{h(1-h/n)}}{\sqrt{n}}.
\end{align}
It is readily verified that~\eqref{equation:relativeerror} vanishes as~$n$ grows if and only if~$h=o(n)$. In this case, the difference between the approximated signal from our algorithm and the true signal is negligible relative to the overall range, for any signal, indicating good approximation.


\begin{theorem}\label{theorem:ApproxMI}
	In the algorithm above we have~$I(W;Q)=n-t-\log|\Gamma|$.
\end{theorem}
\begin{proof}
	Let~$\boldq$ be a query in the support of~$Q$. Since~$\boldu'$ is a deterministic function of~$\boldq$, revealing~$\boldq$ to the user identifies the coset~$V\oplus(\boldw\oplus\boldg)=V\oplus\boldu'$. That is, by conditioning on~$Q=\boldq$, the user learns that~$\boldw$ belongs to the coset~$V\oplus\boldw=V\oplus\boldu'\oplus \boldg$ for some~$\boldg$ that is unknown to the user. Hence, due to Lemma~\ref{lemma:cosetUniqueness}, from the perspective of the user~$\boldw$ is equally likely to be any vector in either of the distinct cosets~$\{ V\oplus\boldu'\oplus\boldg \}_{\boldg\in\Gamma}$, and therefore~$W\vert Q=\boldq$ is uniform over a set of size~$2^t\cdot|\Gamma|$, whose entropy is~$t+\log|\Gamma|$. It follows that $H(W|Q) = t + \log|\Gamma|$ and $I(W;Q)=n-t-\log|\Gamma|$.
\end{proof}

To find the improvement over the exact-computation protocol, we compute the size of~$\Gamma$. The proof of the following lemma is straightforward.
\begin{lemma}\label{lemma:valueofGamma}
	If~$h<\frac{n}{2t}$ then~$|\Gamma|=\binom{n}{h}$. Else, let~$E=\{ (e_1,\ldots,e_t)\in\bZ^t\vert \sum_{i=1}^te_i\le h\text{ and }0\le e_i<\frac{n}{2t} \}$, and
	\begin{align*}
		|\Gamma|=\sum_{(e_1,\ldots,e_t)\in E}\binom{n/t}{e_1}\cdots\binom{n/t}{e_t}.
	\end{align*}
\end{lemma}

If, for instance, $h =  n/\log(n)=o(n)$ and $t = n/10$, then~$|\Gamma|=\sum_E\prod_i \binom{10}{e_i}\ge \binom{n/10}{n/\log(n)}\cdot10^{n/\log(n)}$ \ifarXiv (since there are~$\binom{t}{h}$ $\{0,1\}$-vectors of length~$t$ and weight~$h$, all are in~$E$) \fi, which results in~$\log|\Gamma|=\Omega(n\log\log n/\log n)$ improvement over~\cite{raviv2022information} with vanishing error. On the other hand, if one settles for a non-vanishing error~$0<\delta<1$ in~\eqref{equation:relativeerror}, then choosing~$\delta<1/\sqrt{10}$ and~$h=\delta^2n$ will similarly guarantee at most~$\delta$ relative error and $\log|\Gamma|=\Omega(n)$ improvement.

\section{Approximate Private Inference Lower Bound}\label{section:lowerbound}
In this section we prove an inverse to Theorem~\ref{theorem:ApproxMI} for the case of \textit{linear} decoding. That is, it is assumed that the function used to approximate~$\boldw\boldx^\intercal$ is a linear combination of the responses~$\{ \boldv_i'\boldx^\intercal \}_{i=1}^{\ell}$, and compute a respective lower bound on~$I(W;Q)$. 

We assume that there exist coefficients~$\alpha_1,\ldots,\alpha_\ell$, which depend on~$Q$ but not on~$X$, with which the server approximates~$\boldw\boldx^\intercal$. Furthermore, it is assumed that the approximation is bounded whenever~$\norm{\boldx}_2$ is bounded. Specifically, we assume there exists a constant~$\epsilon>0$ such that
\begin{align}\label{equation:approximateDef}
	|\boldw\boldx^\intercal-\textstyle\sum_{i=1}^{\ell}\alpha_i\boldv_i\boldx^\intercal|\le \epsilon
\end{align}
for every~$\boldx$ such that~$\norm{\boldx}_2\le 1$. In the following lemma it is shown that~\eqref{equation:approximateDef} is equivalent to having a vector~$\boldw'=\sum_{i=1}^{\ell}\alpha_i\boldv_i'$ close to~$\boldw$ in~$\ell_2$-norm.
\begin{lemma}\label{lemma:closew'}
For~$\boldw,\boldw'\in\bR^n$ and~$\epsilon>0$, if~$|\boldx(\boldw-\boldw')^\intercal|\le\epsilon$ for every~$\boldx$ with~$\norm{\boldx}_2\le 1$, then~$\norm{\boldw-\boldw'}_2\le \epsilon\cdot\sqrt{\ln(n)+1}$.
\end{lemma}
\begin{proof}
	For~$\boldr=(r_i)_{i=1}^n=\boldw-\boldw'$, we show that for every~$u\in[n]$, there are at most~$u-1$ entries~$r_j$ of~$\boldr$ such that~$|r_j|>\frac{1}{\sqrt{u}}\epsilon$. Assume for contradiction that there exists a set~$L\subseteq[n]$ of size~$u$ that $|r_j|>\frac{1}{\sqrt{u}}\epsilon$ for every~$j\in L$. Then, for the vector
	\begin{align*}
		\boldx=\begin{cases}
			\tfrac{1}{\sqrt{u}}\cdot \operatorname{sign}(r_j) & \text{if }j\in L\\
			0 & \text{otherwise},
		\end{cases}
	\end{align*}
	which satisfies~$\norm{\boldx}_2=1$, we have~$|\boldx\boldr^\intercal|=|\frac{1}{\sqrt{u}}\sum_{j\in L}|r_j||=\frac{1}{\sqrt{u}}\sum_{j\in L}|r_j|>\frac{1}{\sqrt{u}}\cdot u \cdot\frac{1}{\sqrt{u}}\cdot\epsilon=\epsilon$, a contradiction. Therefore, applying this statement for~$u=1,2,\ldots$, we have that
	\begin{itemize}
		\item All entries of~$\boldr$ satisfy~$|r_j|\le \epsilon$.
		\item There is at most one entry such that~$|r_j|>\epsilon/\sqrt{2}$.
		\item There are at most two entries such that~$|r_j|>\epsilon/\sqrt{3}$.
		\item $\ldots$
	\end{itemize}
	Hence, assuming without loss of generality that~$|r_1|\ge |r_2|\ge \ldots \ge |r_n|$, in the worst case we have
	\begin{align*}
		\epsilon &\ge |r_1| >\epsilon/\sqrt{2}\\
		\epsilon/\sqrt{2} &\ge |r_2| >\epsilon/\sqrt{3}\\
		\epsilon/\sqrt{3} &\ge |r_3| >\epsilon/\sqrt{4}\quad\ldots
	\end{align*}
	and therefore, $\norm{\boldr}_2\le \sqrt{\epsilon^2\sum_{i=1}^{n}\frac{1}{i}}=\epsilon\cdot \sqrt{H_n}$,
	where~$H_n$ is the~$n$'th harmonic number. The claim then follows from a known bound on~$H_n$~\cite{havil2003gamma}. 
\end{proof}


For any algorithm which satisfies~\eqref{equation:approximateDef} we show a lower bound on~$I(W;Q)$ by showing an upper bound on~$H(W\vert Q)$. As done in~\cite{raviv2022information}, we first bound the support size of~$W\vert Q=\boldq$ for an arbitrary~$\boldq$. The latter is given by bounding the number of~$\{\pm 1\}^n$ vectors in the~$\epsilon'$-neighborhood of any~$\ell$-dimensional subspace of~$\bR^n$, where~$\epsilon'=\epsilon\sqrt{\ln(n)+1}$. In the following let~$R$ be an~$\ell$-subspace of~$\bR^n$, and let~$\cN(R,\epsilon')=\{ \boldy\in\bR^n\vert d_2(\boldy,R)\le\epsilon' \}$ be its~$\epsilon'$-neighborhood, where~$d_2(\boldy,R)\triangleq\inf_{\boldr\in R}d_2(\boldy,\boldr)$. The proof of the following lemma is based on intersection properties of subspaces with orthants of~$\bR^n$, and is given in \ifarXiv the appendix\else~\cite{arXiv}\fi.

\begin{lemma}\label{lemma:neighborhoodbound}
	For an~$\ell$-dimensional subspace~$R$ of~$\bR^n$ we have
        \begin{align*}
            |\cN(R,\epsilon')\cap\{\pm1 \}^n|&\le \left( 2\cdot \sum_{j=0}^{\ell-1}\binom{n-1}{j} \right)\left( \sum_{j=0}^{\floor{\epsilon'^2}}\binom{n}{j} \right)\\
            &\triangleq \mu(\ell,n,\epsilon').
        \end{align*}
\end{lemma}

This gives rise to the following bound.

\begin{theorem}
    Assuming linear decoding, we have that~$I(W;Q)\ge n-\log\mu(\ell,n,\epsilon')$.
\end{theorem}

\begin{proof}
    Since the~$\boldv_i$'s are a deterministic function of~$\boldq$, it follows that~$H(W|Q)=H(W,\{\boldv_i\}_{i=1}^\ell\vert Q)$. Therefore, it suffices to bound $H(W,\{\boldv_i\}_{i=1}^\ell\vert Q=\boldq)$ for an arbitrary~$\boldq$. 
    
    Fix any~$\boldq$, which determines the vectors~$\{\boldv_i\}_{i=1}^\ell$, and let~$R=\operatorname{Span}_\bR\{\boldv_i\}_{i=1}^\ell$. According to~\eqref{equation:approximateDef} and Lemma~\ref{lemma:closew'}, it follows that $\boldw\in\cN(R,\epsilon')$, i.e., $\boldw$ is a~$\{\pm1\}$-vector in the~$\epsilon'$-environment of $R$. Therefore, the support size of~$W,\{\boldv_i\}_{i=1}^\ell\vert Q=\boldq$ cannot be larger than the maximum number of~$\{\pm1\}$-vectors in an~$\epsilon'$ environment of an~$\ell$-dimensional subspace. 

    According to a known exercise~(e.g., \cite[Thm.~2.6.4]{CovThom06}), the entropy of any random variable is bounded from above by the logarithm of its support size, and therefore
    \begin{align*}
        &H(W\vert Q)=H(W,\{\boldv_i\}_{i=1}^\ell\vert Q)\\
        &= \sum_{\boldq\in\operatorname{Supp}(Q)}\Pr(Q=\boldq)H(W,\{\boldv_i\}_{i=1}^\ell|Q=\boldq)\\
        &\le \log|\operatorname{Supp(W,\{\boldv_i\}_{i=1}^\ell\vert Q=\boldq)}|\sum_{q\in\operatorname{Supp(Q)}}\Pr(Q=\boldq)\\
        &\le \log \mu(\ell,n,\epsilon'),
    \end{align*}
    where the last transition follows from Lemma~\ref{lemma:neighborhoodbound}. 
    \end{proof}

    Unfortunately, the above bound trivializes for a certain range of~$\ell$ values. In what follows we analyze the bound, and compare it to the mutual information guaranteed by our algorithm in Section~\ref{section:breaking}. We are going to utilize a known result (see \ifarXiv the appendix\else~\cite{arXiv}\fi) about the binomial sum: $\sum_{k=0}^{m} \binom{n}{k} \le (\frac{e\cdot n}{m})^m$, where $e$ is the base of the natural logarithm. Assuming that~$\ell$ is non-constant and that~$n$ is reasonably large (specifically, that~$e\le \epsilon'^2)$, it follows that
   \begin{align}\label{equation:boundIWQ}
             \mu(\ell,n,\epsilon') &\le 2 \left(\frac{e(n-1)}{\ell-1}\right)^{\ell-1}\left( \frac{e n}
            {\epsilon'^2} \right)^{\epsilon'^2} \nonumber\\ 
            &\le 2n^{\ell+\epsilon'^2}\le 2n^{\ell+\epsilon^2\log(n)}\text{ and hence,}\nonumber\\
            I(W;Q)&\ge n - \ell\log(n)-\epsilon^2\log^2(n)-1.
        \end{align}

Clearly, \eqref{equation:boundIWQ} trivializes for~$\ell\ge n/\log(n)$, but is meaningful for other parameter regimes. For instance, applying the scheme in Section~\ref{section:breaking} with~$t=n^{1-\delta}$ for some~$0<\delta<1$ and any constant~$h$ (which results in approximation of~$2\sqrt{h(1-h/n)}$ by Theorem~\ref{theorem:errorbound}) results in
\begin{align*}
    I(W;Q)\approx n-n^{1-\delta}-h\log(n)
\end{align*}
whereas by~\eqref{equation:boundIWQ} with~$\epsilon\le 2\sqrt{h}$ and~$\ell=t=n^{1-\delta}$, it can potentially be lowered as much as
\begin{align*}
    I(W;Q)\ge n-n^{1-\delta}\log(n)-4h\log^2(n).
\end{align*}
Closing this gap, as well as extending it to~$\ell\ge n/\log(n)$, is an interesting problem for future work. The crux seems to be finding suitable subspaces which contain ``many''~$\{\pm1\}$ vectors in their~$\epsilon'$-environment, or improving the bound in Lemma~\ref{lemma:neighborhoodbound}.

\printbibliography

@inproceedings{raviv2022information,
  title={Information Theoretic Private Inference in Quantized Models},
  author={Raviv, Netanel and Bitar, Rawad and Yaakobi, Eitan},
  booktitle={2022 IEEE International Symposium on Information Theory (ISIT)},
  pages={1641--1646},
  year={2022},
  organization={IEEE}
}

@inproceedings{arXiv,
%   title={Approximate Private Inference in Quantized Models},
%   author={Deng, Zirui and Raviv, Netanel},
%   booktitle={arxiv: XXXX:XXXX [CS]},
%   year={2023}
% }

@article{djokovic2014some,
	title={Some new orders of Hadamard and Skew-Hadamard matrices},
	author={Dokovic, Dragomir Z and Golubitsky, Oleg and Kotsireas, Ilias S},
	journal={Journal of combinatorial designs},
	volume={22},
	number={6},
	pages={270--277},
	year={2014},
	publisher={Wiley Online Library}
}

@inproceedings{hubara2016binarized,
	title={Binarized neural networks},
	author={Hubara, Itay and Courbariaux, Matthieu and Soudry, Daniel and El-Yaniv, Ran and Bengio, Yoshua},
	booktitle={Proceedings of the 30th International Conference on Neural Information Processing Systems},
	pages={4114--4122},
	year={2016}
}

@inproceedings{teerapittayanon2017distributed,
  title={Distributed deep neural networks over the cloud, the edge and end devices},
  author={Teerapittayanon, Surat and McDanel, Bradley and Kung, Hsiang-Tsung},
  booktitle={2017 IEEE 37th international conference on distributed computing systems (ICDCS)},
  pages={328--339},
  year={2017},
  organization={IEEE}
}

@inproceedings{boemer2020mp2ml,
  title={MP2ML: A mixed-protocol machine learning framework for private inference},
  author={Boemer, Fabian and Cammarota, Rosario and Demmler, Daniel and Schneider, Thomas and Yalame, Hossein},
  booktitle={Proceedings of the 15th International Conference on Availability, Reliability and Security},
  pages={1--10},
  year={2020}
}

@inproceedings{reagen2021cheetah,
  title={Cheetah: Optimizing and accelerating homomorphic encryption for private inference},
  author={Reagen, Brandon and Choi, Woo-Seok and Ko, Yeongil and Lee, Vincent T and Lee, Hsien-Hsin S and Wei, Gu-Yeon and Brooks, David},
  booktitle={2021 IEEE International Symposium on High-Performance Computer Architecture (HPCA)},
  pages={26--39},
  year={2021},
  organization={IEEE}
}

@inproceedings{mohassel2017secureml,
  title={Secureml: A system for scalable privacy-preserving machine learning},
  author={Mohassel, Payman and Zhang, Yupeng},
  booktitle={2017 IEEE symposium on security and privacy (SP)},
  pages={19--38},
  year={2017},
  organization={IEEE}
}

@inproceedings{juvekar2018gazelle,
  title={$\{$GAZELLE$\}$: A low latency framework for secure neural network inference},
  author={Juvekar, Chiraag and Vaikuntanathan, Vinod and Chandrakasan, Anantha},
  booktitle={27th USENIX Security Symposium (USENIX Security 18)},
  pages={1651--1669},
  year={2018}
}

@inproceedings{gilad2016cryptonets,
	title={Cryptonets: Applying neural networks to encrypted data with high throughput and accuracy},
	author={Gilad-Bachrach, Ran and Dowlin, Nathan and Laine, Kim and Lauter, Kristin and Naehrig, Michael and Wernsing, John},
	booktitle={International Conference on Machine Learning},
	pages={201--210},
	year={2016},
	organization={PMLR}
}

@BOOK{CovThom06,
	title = {Elements of Information Theory},
	publisher = {Wiley},
	year = {2006},
	author = {T. M. Cover and J. A. Thomas},
	address = {New-York},
	edition = {2nd}
}

@inproceedings{raviv2021enhancing,
  title={Enhancing Robustness of Neural Networks through Fourier Stabilization},
  author={Raviv, Netanel and Kelley, Aidan and Guo, Minzhe and Vorobeychik, Yevgeniy},
  booktitle={International Conference on Machine Learning},
  pages={8880--8889},
  year={2021},
  organization={PMLR}
}

@article{havil2003gamma,
  title={Gamma: exploring Euler's constant},
  author={Havil, Julian},
  journal={The Australian Mathematical Society},
  pages={250},
  year={2003}
}

@inproceedings{mcdanel2017embedded,
	title={Embedded Binarized Neural Networks},
	author={McDanel, Bradley and Teerapittayanon, Surat and Kung, HT},
	booktitle={Proceedings of the 2017 International Conference on Embedded Wireless Systems and Networks},
	pages={168--173},
	year={2017}
}
\clearpage
\appendices

\ifarXiv
    \section{Omitted proofs and derivations}\label{section:omitted}
    \begin{proof}[Proof of Lemma~\ref{lemma:Kronecker}]
    	We demonstrate this via an example; the general case follows by straightforward generalization. Let~$n=12$,~$t=4$, and 
    	\begin{align*}
    		&S_1=\{1,2,3\}\quad S_2=\{ 4,5,6 \}\\ &S_3=\{ 7,8,9 \}\quad S_4=\{ 10,11,12\},
    	\end{align*}
    which implies that
    \begin{align*}
    	\1_{S_1} &= [-1,-1,-1,\phantom{-}1,\phantom{-}1,\phantom{-}1,\phantom{-}1,\phantom{-}1,\phantom{-}1,\phantom{-}1,\phantom{-}1,\phantom{-}1],\\
    	\1_{S_2} &= [\phantom{-}1,\phantom{-}1,\phantom{-}1,-1,-1,-1,\phantom{-}1,\phantom{-}1,\phantom{-}1,\phantom{-}1,\phantom{-}1,\phantom{-}1],\\
    	\1_{S_3} &= [\phantom{-}1,\phantom{-}1,\phantom{-}1,\phantom{-}1,\phantom{-}1,\phantom{-}1,-1,-1,-1,\phantom{-}1,\phantom{-}1,\phantom{-}1],\\
    	\1_{S_4} &= [\phantom{-}1,\phantom{-}1,\phantom{-}1,\phantom{-}1,\phantom{-}1,\phantom{-}1,\phantom{-}1,\phantom{-}1,\phantom{-}1,-1,-1,-1].
    \end{align*}
    	Then, choosing the following~$L$
    	\begin{align*}
    		L=\begin{bmatrix}
    			1 & \phantom{-}1 & \phantom{-}1 & \phantom{-}1 \\
    			1 & -1 & \phantom{-}1 & -1 \\
    			1 & \phantom{-}1 & -1 & -1 \\
    			1 & -1 & -1 & \phantom{-}1 
    		\end{bmatrix},
    	\end{align*}
    	which is invertible over~$\bR$ as a Hadamard matrix of order~$4$, we have
    	\setcounter{MaxMatrixCols}{20}
    	\begin{align*}
    		C&=L\odot \begin{bmatrix}
    			\1_{S_1}\\
    			\1_{S_2}\\
    			\1_{S_3}\\
    			\1_{S_4}
    		\end{bmatrix}=
    		\begin{bmatrix}
    			\1\\
    			\1_{S_2}\oplus \1_{S_4}\\
    			\1_{S_3}\oplus \1_{S_4}\\
    			\1_{S_2}\oplus \1_{S_3}
    		\end{bmatrix}=\\
    &=\begin{bmatrix}
    	1\cdot\1_{3} & \phantom{-}1\cdot\1_{3} & \phantom{-}1\cdot\1_{3} & \phantom{-}1\cdot\1_{3} \\
    	1\cdot\1_{3} & -1\cdot\1_{3} & \phantom{-}1\cdot\1_{3} & -1\cdot\1_{3} \\
    	1\cdot\1_{3} & \phantom{-}1\cdot\1_{3} & -1\cdot\1_{3} & -1\cdot\1_{3} \\
    	1\cdot\1_{3} & -1\cdot\1_{3} & -1\cdot\1_{3} & \phantom{-}1\cdot\1_{3}
    \end{bmatrix}=L\otimes \1_3.\qedhere
    	\end{align*}
    \end{proof}
    
    \begin{proof}[Derivation of \eqref{equation:ww'dist}]
        \begin{align*}
    	&\norm{\boldw-\boldw'}_2^2=\norm{\boldw(I-\tfrac{t}{n}U'\operatorname{diag}(J_{n/t})U')}_2^2\nonumber\\
    	&= \boldw (I-\tfrac{t}{n}U'\operatorname{diag}(J_{n/t})U')(I-\tfrac{t}{n}U'\operatorname{diag}(J_{n/t})U')^\intercal\boldw^\intercal\nonumber\\
    	&= \boldw\boldw^\intercal-\tfrac{2t}{n}\cdot \boldw U' \operatorname{diag}(J_{n/t})U' \boldw^\intercal+\nonumber\\
            &\phantom{=}+\tfrac{t^2}{n^2}\cdot \boldw U' \operatorname{diag}(J_{n/t})^2U'\boldw^\intercal\nonumber\\
    	&= n-\tfrac{2t}{n}\cdot \boldw U' \operatorname{diag}(J_{n/t})U' \boldw^\intercal+\nonumber\\
            &\phantom{=}+\tfrac{t^2}{n^2}\cdot\tfrac{n}{t} \cdot \boldw U' \operatorname{diag}(J_{n/t})U'\boldw^\intercal\nonumber\\
    	&= n-\tfrac{t}{n}\cdot \boldw U' \operatorname{diag}(J_{n/t})U' \boldw^\intercal\nonumber\\
    	&= n-\tfrac{t}{n}(\boldw\oplus\boldu')\operatorname{diag}(J_{n/t})(\boldw\oplus\boldu')^\intercal.\qedhere
    \end{align*}
    \end{proof}
    
    \begin{proof}[Derivation of~\eqref{equation:fmincont}]
        Since~$\operatorname{diag}(J_{n/t})_{i,j}$ equals~$1$ if~$i,j\in S_p$ for some~$p\in[t]$, and zero otherwise, it follows that
    \begin{align*}
    \eqref{equation:ww'dist}	&= n-\tfrac{t}{n}\sum_{i,j\in[n]}(\boldw\oplus\boldu')_i(\boldw\oplus\boldu')_j \operatorname{diag}(J_{n/t})_{i,j}\nonumber\\
    &= n-\tfrac{t}{n}\sum_{p=1}^t\sum_{i,j\in S_p}(\boldw\oplus\boldu')_i(\boldw\oplus\boldu')_j\nonumber\\
    							&= n-\tfrac{t}{n}\sum_{p=1}^{t}\left( \sum_{i\in S_p}(\boldw\oplus\boldu')_i \right)^2\nonumber\\
    							&=n-\tfrac{t}{n}\sum_{p=1}^{t}(n/t-2w_H(\boldw\oplus\boldu'\vert_{S_p}))^2\nonumber\\
    							& = n-\tfrac{t}{n}\left( \tfrac{n^2}{t}-\tfrac{4n}{t}\sum_{p=1}^{t}w_H(\boldw\oplus\boldu'\vert_{S_p})\right.\nonumber\\
                                    &\phantom{=====} +\left.4\sum_{p=1}^{t}w_H(\boldw\oplus\boldu'\vert_{S_p})^2 \right)\nonumber\\
           &=4d_H(\boldw,\boldu')-\tfrac{4t}{n}\sum_{p=1}^tw_H(\boldw\oplus\boldu'\vert_{S_p})^2.\qedhere
    \end{align*}
    \end{proof}
    
    \begin{proof}[Proof of Theorem~\ref{theorem:errorbound}]
    	By the definition of $\boldu'$, we have that~$(\boldg\oplus\boldw)\oplus\boldu'\in V$ (or alternatively, that~$V\oplus(\boldg\oplus\boldw)=V\oplus \boldu'$) for any choice of~$\boldg$. Therefore, according to Lemma~\ref{lemma:upupp}, the closest vector in Euclidean distance to~$\boldw$ in~$V\oplus (\boldg\oplus\boldw)$ is~$\boldw'=\boldalpha_{\boldu'}^\text{min} C U'=\sum_{i=1}^{t}\alpha_{\boldu',i}^\text{min}\boldv_i'$. Further, according to Lemma~\ref{lemma:resultingApproximation} and Lemma~\ref{lemma:upupp}, we have
    	\begin{align}\label{equation:ww'}
    		&\norm{\boldw-\boldw'}_2=\nonumber\\
      &=2\sqrt{d_H(\boldw,\boldg\oplus\boldw)-\tfrac{t}{n}\textstyle\sum_{p=1}^tw_H(\boldw\oplus\boldg\oplus\boldw\vert_{S_p})^2}\nonumber\\
    		&=2\sqrt{w_H(\boldg)-\tfrac{t}{n}\textstyle\sum_{p=1}^tw_H(\boldg\vert_{S_p})^2}
    	\end{align}
    	To provide a worst-case bound, we find the maximum of~\eqref{equation:ww'} over all~$\boldg\in\Gamma$ by solving the following integer programming problem, in which~$d_i$ is a shorthand notation for~$w_H(\boldg\vert_{S_i})$.
    	\begin{align}\label{equation:integerProg}
    		&\text{Maximize: } \textstyle\sum_{i=1}^{t}\left[d_i-\tfrac{t}{n}d_i^2  \right]\\
    		&\text{Subject to: } \textstyle\sum_{i=1}^{t} d_i\le h \text{ and }0\le d_i< \tfrac{n}{2t}\text{ for all }i\in[t].\nonumber
    	\end{align}
    	Clearly, any feasible solution to the above problem satisfies~$\sum_{i=1}^{t}d_i=h'$ for some integer~$h'\le h$. Hence, we bound the above objective value under the constraint~$\sum_{i=1}^{t}d_i=h'$ instead of the constraint~$\sum_{i=1}^{t}d_i\le h$, and then maximize over all~$h'$. Further, note that maximizing~$\sum_{i=1}^t [d_i-\frac{t}{n}d_i^2]$ under the constraint~$\sum_{i=1}^{t}d_i=h'$ is equivalent to maximizing~$h'-\frac{t}{n}\sum_{i=1}^{t}d_i^2$, and hence equivalent to minimizing~$\sum_{i=1}^{t}d_i^2$. Similarly, removing constraints and solving the problem over~$\bR$ can only improve (that is, reduce) the resulting minimum, and hence to get a bound it suffices to solve		
    	\begin{align}\label{equation:LagrangeMultipliers}
    		&\text{Minimize: } \textstyle\sum_{i=1}^{t}d_i^2\nonumber\\
    		&\text{Subject to: } \textstyle\sum_{i=1}^{t} d_i= h'
    	\end{align}
    	over~$\bR$.	Straightforward computations using Lagrange multipliers (which are given separately below) show that the optimum of this objective equals~$h'^2/t$, and it is obtained at the point~$(d_1,\ldots,d_t)=(h'/t)\cdot\1$. Hence, the solution of the \textit{integer} minimization problem 
    	\begin{align*}
    		&\text{Minimize: } \textstyle\sum_{i=1}^{t}d_i^2\nonumber\\
    		&\text{Subject to: } \textstyle\sum_{i=1}^{t} d_i= h' \text{ and }0\le d_i< \tfrac{n}{2t} \text{ for all }i\in[t]
    	\end{align*}
    	is bounded from below by~$h'^2/t$, and thus the solution for the integer problem
    	\begin{align*}
    		&\text{Maximize: } \textstyle\sum_{i=1}^{t}[d_i-\tfrac{t}{n}d_i^2]\nonumber\\
    		&\text{Subject to: } \textstyle\sum_{i=1}^{t} d_i= h' \text{ and }0\le d_i< \tfrac{n}{2t}\text{ for all }i\in[t]
    	\end{align*}
    	is bounded from above by~$h'-\frac{t}{n}\cdot\frac{h'^2}{t}=h'-h'^2/n$. Since the function~$x-x^2/n$ has a positive derivative in the range~$0\le x\le n/2$, and since~$h\le n/2$ is a necessary condition, it follows that~$h'-h'^2/n$ is \textit{increasing} with~$h'$, and hence the optimal solution to~\eqref{equation:integerProg} is bounded from above by~$h(1-h/n)$. Therefore, we have that $$\eqref{equation:ww'}\le 2\sqrt{h(1-h/n)},$$ and the proof is concluded by using Lemma~\ref{lemma:CauchySchwartz}.
    \end{proof}
    
    \begin{proof}[Solution to~\eqref{equation:LagrangeMultipliers}] 
        We wish to find the minimum of the function~$h(d_1,\ldots,d_t)=\sum_{i=1}^t d_i^2$ under the constraint~$\sum_{i=1}^td_i=h'$. Using the Lagrange multipliers method, we define~$g'(d_1,\ldots,d_t,\lambda)=\sum_{i=1}^td_i^2-\lambda(\sum_{i=1}^td_i-h')$, compute the gradient
        \begin{align*}
            \nabla g'(d_1,\ldots,d_t,\lambda)=(2d_1-\lambda,\ldots,2d_t-\lambda,\sum_{i=1}^td_i-h'),
        \end{align*}
        and find that the gradient equals zero when~$d_i=\lambda/2$ for all~$i$. In turn, this implies that in the optimum we have that~$\sum_{i=1}^t d_i=t\lambda/2=h'$. Hence,~$\lambda=2h'/t$, and the optimum obtained at the respective point~$(h'/t,\ldots,h'/t)$ equals~$\sum_{i=1}^t (h'/t)^2 = h'^2/t$.
    \end{proof}
    
    \begin{proof}[Proof of known result used in Section \ref{section:lowerbound}]
        Let $f(m, n) = \sum_{k=0}^{m} \binom{n}{k}$ and $g(m, n) =  (\frac{e\cdot n}{m})^m$. We prove $f(m,n) \le g(m,n)$ by induction with respect to $n$. For $m = 1$ we have $f(m,n) = 1 + n < e\cdot n = g(m,n)$, and for $m = n$ we have $f(m,n) = 2^n < e^n = g(m,n)$. Now, using the fact that $\binom{n}{k} + \binom{n}{k + 1} = \binom{n + 1}{k + 1}$ for all $k$ such that $0 \le k \le n - 1$, we have $f(m+1, n+1) = \sum_{k=1}^{m+1} \binom{n+1}{k} - 1 = \sum_{k=1}^{m+1} \binom{n}{k-1} + \sum_{k=1}^{m+1} \binom{n}{k} - 1 = f(m,n) + f(m+1,n)$ for all $m$ such that $1 \le m \le n-1$. Thus, in order to show that $f(m+1,n+1) \le g(m+1,n+1)$, it suffices to show that $g(m+1,n+1) \ge g(m,n) + g(m+1,n)$. Now,
        \begin{align*}
            g(m+1,n+1) &\ge g(m,n) + g(m+1,n) \\
             \left(\frac{e(n+1)}{m+1}\right)^{m+1} &\ge \left(\frac{e n}{m}\right)^m + \left(\frac{e n}{m+1}\right)^{m+1} \\ 
             e\left(\frac{n+1}{m+1}\right)^{m+1} &\ge \left(\frac{n}{m}\right)^m + e\left(\frac{n}{m+1}\right)^{m+1}.
        \end{align*}
        Since
        \begin{align*}
            &e\left(\frac{n+1}{m+1}\right)^{m+1} = e\left(\frac{n}{m+1} + \frac{1}{m+1}\right)^{m+1} \\ &\ge e\left(\frac{n}{m+1}\right)^{m+1} + e(m+1)\left(\frac{n}{m+1}\right)^m \frac{1}{m+1},
        \end{align*}
        it suffices to show that 
        \begin{align*}
            e\left(\frac{n}{m+1}\right)^m &\ge \left(\frac{n}{m}\right)^m \\
            e &\ge \left(\frac{m+1}{m}\right)^m \\
            e^{\frac{1}{m}} &\ge \frac{m+1}{m}.
        \end{align*}
        Recalling the Taylor approximation $e^{\frac{1}{m}} = 1 + \frac{1}{m} + \frac{1}{2m^2} + \ldots$ of~$e^x$ near zero concludes the proof.
    \end{proof}
    
    
    \section{Proof of Lemma~\ref{lemma:neighborhoodbound}}\label{section:Kens}
    To bound $|\cN(R,\epsilon')\cap\{\pm1 \}^n|$, we recall the definition of orthants.
    
    \begin{definition}\label{def:orthant}
        An orthant in $\bR^n$ is a subset obtained through requiring each coordinate to be either nonnegative or nonpositive. Specifically, it is defined by a system of inequalities $a_1x_1 \ge 0, a_2x_2 \ge 0, \ldots, a_nx_n \ge 0$, where each $a_i$ is either $+1$ or $-1$, and $\bolda = (a_1, a_2, \ldots, a_n) \in \{\pm1 \}^n$ is called the signature of the orthant.
    \end{definition}
    
    We consider the union of orthants $U = \cup_{i=1}^k Q_{\bolda^{(i)}}$ that are intersected by~$R$ non-trivially (i.e., other than at the origin), where~$k$ is the number of orthants and the $\bolda^{(i)}$'s are their signatures. Since $\cN(R,\epsilon')\subseteq \cN(U,\epsilon')$, it suffices to upper bound $|\cN(U,\epsilon')\cap\{\pm1 \}^n|$; this will be done by bounding $k$, bounding the number of $\{\pm 1\}^n$ vectors in the $\epsilon'$-neighborhood of an individual orthant, and then multiplying the two.
    
    In order to count the maximum number of orthants in $\bR^n$ that are intersected by an~$\ell$-dimensional subspace, we first establish an equivalent problem of counting the maximum number of cells that~$n$ linear hyperplanes\footnote{I.e.,~$(\ell-1)$-dimensional subspaces in~$\bR^\ell$ which contain the origin.} can split $\bR^{\ell}$ into. This equivalence is a known exercise which appears in several online sources and books; the proof is provided herein for completeness. We represent~$n$ hyperplanes as homogeneous linear maps $f_1, \ldots, f_n$, with $f_i(\boldx)=\boldx\boldw_i^\intercal$ for some~$\boldw_i$'s, and the hyperplanes are their kernels. A cell is a subset of $\bR^{\ell}$ with a specified sign for each map.
    
    \begin{proposition}\label{prop:equivalence}
        Let~$R$ be an $\ell$-dimensional subspace of~$\bR^n$, and let~$f_1,\ldots,f_n$ be any linear maps\footnote{For example, fix a matrix~$M$ whose rows are a basis of~$R$, and then~$f_j(\boldx)=(\boldx M)_j$.} which define~$R$, i.e.,~$R=f(\bR^\ell)$, where~$f=(f_1,\ldots,f_n)$, and~$f_i:\bR^\ell\to\bR$ for all~$i$. Then, the number of orthants in $\bR^n$ intersected by~$R$ is equal to the number of cells induced by~$f_1,\ldots,f_n$ in~$\bR^\ell$. Conversely, given~$n$ hyperplanes in~$\bR^\ell$ defined by the kernels of (linear and homogeneous) ~$f_1,\ldots,f_n$, the object~$f(\bR^\ell)$, where~$f(\boldx)=(f_1(\boldx),\ldots,f_n(\boldx))$, is a subspace of~$\bR^n$ of dimension at most~$\ell$. Consequently, the number of cells induced by the~$f_i$'s is equal to the number of orthants intersected by~$R$.
    \end{proposition}
    
    
    It follows from Proposition~\ref{prop:equivalence} that the maximum number of orthants that an $\ell$-dimensional subspace of~$\bR^n$ can intersect is equal to the maximum number of cells that~$n$ hyperplanes in~$\bR^\ell$ can induce. We shall now focus on the latter problem, and denote by $F(n, \ell)$ the maximum number of cells created by $n$ hyperplanes in $\bR^\ell$. 
    
    \begin{lemma}\label{lemma:numberoforthants}
        $F(n,\ell) \le 2\cdot \sum_{j=0}^{\ell-1}\binom{n-1}{j}$.
    \end{lemma}
    
    \begin{proof}
        Suppose we have $n$ hyperplanes $P_1, \ldots, P_n$ in $\bR^\ell$, and we would like to add some hyperplane $P'$ and count how many additional cells are created. This is equivalent to counting how many times~$P'$ splits an existing cell into two, which is done by considering the cells induced in $P'$ by the hyperplanes $P_1 \cap P', \ldots, P_n \cap P'$. We can then tell that the new cells in $\bR^\ell$ created via adding $P'$ correspond to the cells $n$ hyperplanes split $\bR^{\ell-1}$ into. Thus we have the relation $F(n + 1, \ell) \le F(n, \ell) + F(n, \ell - 1)$, with initial conditions $F(1, \ell) = 2$ and $F(n, 1) = 2$. Finally by induction we are able to prove that $k = F(n, \ell) \le 2\cdot \sum_{j=0}^{\ell-1}\binom{n-1}{j}$.
    \end{proof}
    
    Next we quantify the number of $\{\pm 1\}^n$ vectors in the $\epsilon'$-neighborhood of any orthant $Q_\bolda$ for signature $\bolda \in \{\pm 1\}^n$, i.e., $|\cN(Q_\bolda,\epsilon')\cap\{\pm1 \}^n|$. 
    
    \begin{lemma}\label{lemma:singleorthantbound}
    $|\cN(Q_\bolda,\epsilon')\cap\{\pm1 \}^n| = \sum_{j=0}^{\floor{\epsilon'^2}}\binom{n}{j}$.
    \end{lemma}
    
    \begin{proof}
        To obtain $|\cN(Q_\bolda,\epsilon')\cap\{\pm1 \}^n|$, we need to characterize the distance between a given $\boldalpha \in \{\pm1 \}^n$ and $Q_\bolda$, which is the distance between $\boldalpha$ and the closest $\boldy \in Q_\bolda$ to $\boldalpha$, i.e. $d_2(Q_\bolda, \boldalpha) = d_2(\boldy, \boldalpha)$. It is easy to see that $\boldy$ lies on the boundaries of $Q_\bolda$. Without loss of generality, suppose $\boldy = (y_1, \ldots, y_m, 0, \ldots, 0)$; that is, $y_i \neq 0$ for all $i \in [m]$, and $y_i = 0$ for all $i \in \{m+1, \ldots, n\}$. We can readily make several observations: 
    
    \begin{enumerate}[label=(\roman*)]
        \item\label{(i)} $\alpha_i \neq a_i$ for all~$i \in \{m+1, \ldots, n\}$. If there exists a $j \in \{m+1, \ldots, n\}$ such that $\alpha_j = a_j$, then we may set $y_j = a_j$ and obtain a closer vector in $Q_\bolda$ to $\boldalpha$. This contradicts the fact that~$\boldy$ is the closest vector in $Q_\bolda$ to~$\boldalpha$.
        \item\label{(ii)} $\alpha_i = a_i$ for all $i \in [m]$. This follows from the same reasoning as Observation~\ref{(i)}. 
        \item\label{(iii)} $|y_i| = 1$ for all $i \in [m]$. This immediately follows from Observation~\ref{(ii)}.
    \end{enumerate}
    
    These observations indicate that~$\boldy$ is a $\{0, \pm 1\}^n$ vector; specifically, $\boldy = (y_1, \ldots, y_n)$, where $y_j = 0$ if $a_j \neq \alpha_j$ and $y_j = \alpha_j$ otherwise, for all $j \in [n]$. It is then straightforward to deduce that $d_2(Q_\bolda, \boldalpha) = d_2(\boldy, \boldalpha) = \sqrt{d_H(\bolda, \boldalpha)}$. 
    This means that $d_2(Q_\bolda, \boldalpha) \leq \epsilon'$ if and only if $ d_H(\bolda, \boldalpha) \leq \epsilon'^2$. Now, all $\boldalpha$'s satisfying $d_H(\bolda, \boldalpha) \leq \epsilon'^2$ lie within a Hamming ball of radius $\epsilon'^2$ centered at $\bolda$. Thus we have $|\cN(Q_\bolda,\epsilon')\cap\{\pm1 \}^n| = |\{\boldalpha \in \{\pm1 \}^n| d_2(Q_\bolda, \boldalpha) \leq \epsilon' \} |=|\{\boldalpha \in \{\pm1 \}^n| d_H(\bolda, \boldalpha) \leq \epsilon'^2 \} | = \sum_{j=0}^{\floor{\epsilon'^2}}\binom{n}{j}$.
    \end{proof}
    
    With all of this we are able to achieve the final result:
        \begin{align*}
                |\cN(R,\epsilon')\cap\{\pm1 \}^n| &\le |\cN(U,\epsilon')\cap\{\pm1 \}^n|
                \\ &\le k\cdot \sum_{j=0}^{\floor{\epsilon'^2}}\binom{n}{j} \\ &\le
                \left( 2\cdot \sum_{j=0}^{\ell-1}\binom{n-1}{j} \right)\left( \sum_{j=0}^{\floor{\epsilon'^2}}\binom{n}{j} \right).
        \end{align*}
\fi
\end{document}